\documentclass[conference]{IEEEtran}\IEEEoverridecommandlockouts
\usepackage{amssymb}
\usepackage{amsmath}
\usepackage{amsfonts}
\usepackage{graphicx}
\usepackage{algorithm}
\usepackage{algorithmic}
\usepackage{epstopdf}
\usepackage{cite}

\newtheorem{lemma}{\textbf{Lemma}}

\newtheorem{proposition}{\textbf{Proposition}}

\begin{document}

\title{Delay Analysis and Optimization in Cache-enabled Multi-Cell Cooperative Networks}
\author{Yaping Sun, Zhiyong Chen, and Hui Liu\\
Cooperative Medianet Innovation Center, Shanghai Jiao Tong University, Shanghai, P. R. China\\
Email: \{yapingsun, zhiyongchen, huiliu\}@sjtu.edu.cn}
\maketitle

\begin{abstract}
Caching at the base stations (BSs) has been widely adopted to reduce the delivery delay and alleviate the backhaul traffic between BSs and the core network. In this paper, we consider a collaborative content caching scheme among BSs in cache-enabled multi-cell cooperative networks, where the requested contents can be obtained from the associated BS, the other collaborative BSs or the core network. 
Novelly, we model the stochastic request traffic and derive a closed form expression for the average delay per request based on multi-class processor sharing queuing theory. We then formulate a cooperative caching optimization problem of minimizing the average delay under the finite cache size  constraint at BSs and show it to be at least NP-complete. Furthermore, we prove it equivalent to the maximization of a monotone submodular function subject to matroid constraints, allowing us to adopt the common greedy algorithm with $1/2$ performance guarantee. A heuristic greedy caching strategy is also developed, achieving a better performance than the conventional greedy solution. Simulation results verify the accuracy of the analytical results and demonstrate the performance gains obtained by our proposed caching scheme.
\end{abstract}

\section{Introduction}
Recent advances of portable devices have stimulated explosive demand for multimedia services, which poses tremendous traffic load in cellular networks. However, the high traffic load mainly consists of duplicately downloading a few number of popular contents. As such, caching the popular contents at the base stations (BSs) is exploited to reduce the duplicate content transmissions between BSs and core networks \cite{content}. In this way, it could not only significantly reduce traffic load but also further improve users' experience of service such as delay.

The cache-enabled cellular network has also triggered new challenges for the network operators. That is where to cache, what to cache and how to cache \cite{5G1}. From an information-theoretic perspective, the authors in \cite{coding} propose a novel coded-caching scheme to minimize the peak traffic load of a single-cell network. In \cite{CC}, the authors analyze the system performance in cache-enabled heterogeneous networks. Authors in \cite{femto} propose approximate caching placement algorithms to minimize the content access delay in FemtoCaching system. 


Different from the above works, we notice that the collaboration among BSs via high-capacity links offers us an additional freedom to devise caching strategies \cite{Mobile3C}. Specifically, the multi-cell cooperation enables users' requests not only to be satisfied directly from its associated BS or the core network via backhaul links but also from other collaborative BSs. Hence, instead of considering the content caching scheme at each BS individually, this paper analyzes the optimal scheme of assigning contents among the BSs to fully take advantage of the finite storage resources and further improve the users' experiences from the global perspective.

Recently, there are several works on the collaborative caching schemes in the literatures \cite{wang1,delay performance,innetwork,online,BScooperation}. The work in \cite{wang1} proposes collaborative algorithms with three objectives: minimizing the inter Internet service providers (ISP) traffic, the intra ISP traffic as well as the total user delays. In its follow-up work \cite{delay performance}, the authors propose a distributed suboptimal algorithm to minimize the sum delay of contents. In \cite{innetwork}, the authors explore the cooperative caching across small BSs and formulate it as a minimization problem of the cost incurred by retrieving files across small BSs and from  the core network. In \cite{online}, the authors propose an online caching algorithm to minimize the overall cost. \cite{BScooperation} maximizes the reward obtained by the cellular network from the uncoded data and the coded data, respectively.

Due to the limit of bandwidth resources, requests arriving at each BS have to queue up for the service. As fetching requested contents from different places (i.e., associated BS, collaborative BSs and the core network) requires transmission delay of different magnitudes, the cooperative caching strategy highly impacts the experienced delay per request including the transmission delay and the waiting delay. However, in all the above related works \cite{innetwork,online,BScooperation,wang1,delay performance}, the sum delay is derived directly from the sum of transmission delay for each content. That is, they have not taken into account the stochastic arrival time of the requests, which proves to be very important in the performance of caching strategies especially in terms of the delay \cite{queuing}.

Therefore, in this paper, we are motivated to adopt a queuing theoretic approach to analyze the average delay per request in a cache-enabled multi-cell cooperative network. We formulate the user request arrival and departure traffics at each BS as a multi-class processor sharing queuing model \cite{queuing3}. A closed form expression for the average delay per request is then derived and a cooperative caching optimization problem is formulated to minimize the average delay. Moreover, we express the optimization problem as the maximization of a nondecreasing submodular function subject to matroid constraints, enabling us to adopt a low-complexity greedy caching algorithm with $1/2$ approximate performance guarantee. A heuristic greedy caching scheme is then proposed, where we not only take into account the content popularity, storage capacity and the user request arrival rate but also the content sizes. Finally, the impacts of the network resources on the performance of the proposed cooperative caching schemes are discussed.

\section{System Model}
In this section, we present the system model and describe the content characteristics. The traffic dynamics of request arrivals and departures in the cache-enabled multi-cell cooperative networks are then described.

\subsection{Network Architecture}
\begin{figure}
\centering
\includegraphics[width=3in,height=2in]{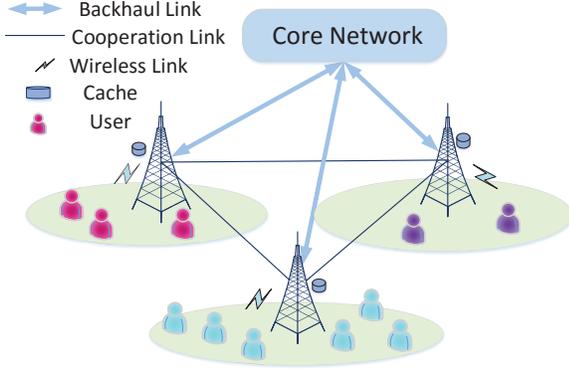}
\caption{Cache-enabled Multi-cell Cooperative Networks}\label{System}
\end{figure}
As illustrated in Fig. \ref{System}, we consider a cache-enabled multi-cell cooperative network with \textit{K} base stations (BSs). Each BS is equipped with a single antenna and is endowed with storage capability, which can store finite number of contents. Let $\mathcal{K} := \{b_k,k=1,2,\cdots,K\}$ denote the set of BSs, and $C_k$ represent the storage capacity in bits of the $k$-th BS. BSs can collaborate with each other via high capacity links, e.g. optical fibre. Meanwhile, each BS is also connected with a core network through backhaul links. Therefore, each requested content can be obtained from the local BS, the collaborative BSs, or the core network, depending on the cache status of all the BSs.

\subsection{Content Characteristics}
We characterize the multimedia contents from two different perspectives, namely, $\{popularity, size\}$. Let $\mathcal{F}$ denote the set of $F$ contents. We use a binary matrix $C:={(c_{k,f})}_{K\times F} \in {\{0,1\}^{K\times F}}$ to denote the cache placement in the multi-cell network, where $c_{k,f}=1$ indicates that the content \textit{f} is cached in BS $b_k$ and $0$ otherwise.
\subsubsection{Content Popularity}
 We use the term \textit{popularity} to describe the probability at which users request a specific content. 
 Let $P_{k,f}$ denote the popularity of content \textit{f} at $b_k$. Considering the overall network, the total popularity for each content, denoted as $p_f = \sum_{k=1}^{K} P_{k,f}$, obeys a Zipf distribution. Specifically,  we have
\begin{equation}\label{pop}
  p_f = \frac{q_f^{-\gamma}}{\sum_{f=1}^{F} q_f^{-\gamma}},
\end{equation}
where $q_f$ denotes the rank of the popularity of content \textit{f} in the descending order and $\gamma\geq0$ is a real constant, characterizing the skewness of the popularity distribution.
\subsubsection{Content Size}
Due to the fact that only a small fraction of multimedia contents have relatively large size and most contents are of limited size \cite{contentexp}, we adopt a long tail style distribution to describe the content size distribution. Based on \cite{kongtao}, the content size is assumed to be an exponential random variable $S_f$ with mean $\bar S$ bits.

\subsection{Traffic Characteristics}
\subsubsection{User Request}
For each BS, we assume the incoming user request stream to be a Poisson process. That is, the request interarrival time is an exponentially distributed random variable. Let $\lambda_k$ denote the aggregate request arrival rate at BS $b_k$. For purpose of deriving closed-form expressions, we consider the stream of requests at each BS to be Independent Reference Model (IRM) \cite{IRM} based on the following assumptions: i) the contents that users request are fixed to the aforementioned content set $\mathcal{F}$; ii) the probability of the request for content $f$ at BS $b_k$, i.e., $P_{k,f}$, is constant and independent of all the past requests. Therefore, the stream of requests at each BS is a sequence of independent and identically distributed (i.i.d) requests. We also assume the request arrival rates and the content preferences across all the cells are different from each other.

\subsubsection{Service Mechanism}
In this paper, we consider each BS serves the stream of user requests sequentially based on first in first out (FIFO) criterion. To serve the request for content $f$ at BS $b_k$, first, BS $b_k$ checks its own cache and delivers content $f$ directly to the user if this content is cached (i.e., $c_{k,f} = 1$), denoted as $Route\  1$. Otherwise, BS $b_k$ fetches the content from a randomly chosen cooperative BS that has cached this content and then transmits it to the user, referred to as $Rroute \ 2$. If all the BSs have not cached this content (i.e., $\sum_{j = 1}^K c_{j,f} = 0$), BS $b_k$ obtains it from the Internet via backhaul link and then delivers it to the user, denoted as $Route\ 3$.

\subsubsection{Service Rate}
Due to the traffic congestion in the core network and the extra transmission delay among cells, without loss of generality, we assume the aggregate transmission bit rates for the above three routes are of diminishing magnitudes. When the request is served via $Route \ 1$, i.e., directly from its affiliated BS, considering the channel is ergodic and the files are always large, taking long enough time to be sent, we assume the average transmission rate is deterministic \cite{servicetime}, denoted as $r_1$ bits/s.
Accordingly, we define the aggregately average transmission rates for $Route \ 2$ and $Route \ 3$ as $r_2=\frac{r_1}{k_2}$ bits/s and $r_3=\frac{r_1}{k_3}$ bits/s, respectively, where $k_3 \geq k_2 \geq 1$.

Moreover, given the content size $S_f$ is exponentially distributed with mean $\bar S$ bits, the corresponding request service time of each route, i.e., $\tau_i = \frac{\bar S}{r_i}$, also follows exponential distribution with mean $\tau_1=\frac{\bar S}{r_1}$ s/request, $\tau_2=\frac{\bar S}{r_2}=k_2\tau_1$ s/request and $\tau_3=\frac{\bar S}{r_3}=k_3\tau_1$ s/request, respectively.

\section{Delay Analysis and Problem Formulation}
In this section, we mainly focus on the analysis of the average delay on per-request basis from the global perspective. Different from the previous works \cite{innetwork,online,BScooperation,wang1,delay performance}, we take into account the arrival rates and the traffic dynamics from a queuing perspective. Thus, we get a closed form expression for the average delay, which facilitates the system design as well as the caching strategy for better network performance. Finally, we formulate the optimization of the cache placement as the minimization of an integer programming problem.

\subsection{Multiclass Processor Sharing Queuing Model}
With respect to the above mentioned three different delivery routes, we formulate the dynamic traffic at each BS as a \textit{multi-class processor sharing queue} (\textit{MPSQ}) \cite{queuing3}. Specifically, we consider each BS as a processor with the user requests as its customers. According to the service mechanism mentioned in Section II-C, for a given content placement of all the BSs, the user requests at each BS can be divided into three groups according to their service routes, i.e., $Route\ 1$, $Route\ 2$ and $Route\ 3$. We assume no group of the user requests has priority over any other and each BS has infinite waiting room.

According to the user requests arrival and service model, the traffic for each request class at any BS can be modeled as an M/M/1 queue. Specifically, Table \ref{table1} presents the request arrival and service rates for each group at BS $b_k$. For denotational convenience, we normalize $P_{k,f}$ as $\bar{P}_{k,f} =\frac{P_{k,f}}{\sum_{f=1}^F P_{k,f}}$.

Corresponding with Table \ref{table1}, we denote with $R_{k,i} = \frac{\lambda_{k,i}}{\lambda_k}$ the probability for requests to be served via the \textit{i-th} route at BS $b_k$. Specifically, $R_{k,1} = \sum_{f=1}^F \bar P_{k,f}c_{k,f}$ indicates the probability for requests satisfied directly from associated BS $b_k$, $R_{k,2} = \sum_{f=1}^F\bar P_{k,f}(1-c_{k,f})\max_{i}c_{i,f}$ measures the probability for requests that have to be served through $Route \ 2$ and $R_{k,3} = \sum_{f=1}^F\bar P_{k,f}\Pi_{i=1}^{K}(1-c_{i,f})$ indicates the local cache loss probability, respectively. 
Accordingly, $\lambda_{k,i}=\lambda_kR_{k,i}$ represents the \textit{i-th} class request arrival rate at BS $b_k$. Obviously, we have $\sum_{i=1}^3 R_{k,i} = 1$ and $\lambda_k = \sum_{i=1}^{3}\lambda_{k,i}$.

\begin{table}
\renewcommand\arraystretch{1.7}
\centering
\caption{TRAFFIC DYNAMIC AT BS $b_k$}\label{table1}
\begin{tabular}{c|l|c}
  \hline
  \textbf{Route} & \textbf{Request Arrival Rate }& \textbf{Service Rate } \\
  \hline
  1 & $\lambda_{k,1} =\lambda_k\frac{\sum_{f=1}^FP_{k,f}c_{k,f}}{\sum_{f=1}^FP_{k,f}}$ & $\mu_1 = \frac{1}{\tau_1}$  \\
  \hline
  2 & $\lambda_{k,2} =\lambda_k\frac{\sum_{f=1}^FP_{k,f}(1-c_{k,f})\max_{i}c_{i,f}}{\sum_{f=1}^FP_{k,f}}$ & $\mu_2 = \frac{1}{k_2\tau_1}$ \\
  \hline
  3 & $\lambda_{k,3} = \lambda_k\frac{\sum_{f=1}^FP_{k,f}\Pi_{i=1}^{K}(1-c_{i,f})}{\sum_{f=1}^FP_{k,f}}$ & $\mu_3 = \frac{1}{k_3\tau_1}$ \\
  \hline
\end{tabular}
\end{table}

%
%

Moreover, we introduce $\rho_k$ as a metric of the traffic intensity at BS $b_k$, defined as:
\begin{equation}\label{stable}
  \rho_k = \sum_{i=1}^3\frac{\lambda_{k,i}}{\mu_i} = \lambda_k\tau_1(R_{k,1} + k_2R_{k,2} + k_3R_{k,3}).
\end{equation}

In this paper, we consider $\rho_k < 1$ as the stability condition. Otherwise, the overall delay will be infinite. In this way, we notice that the traffic intensity at BS $b_k$ is simultaneously related to the cache hit probability, request arrival rates and the transmission capacities for the three routes.

\subsection{Average Delay in A Single Cell}
We first focus on the delay analysis in a single cell with BS $b_k$ and have the following proposition.
\begin{proposition}
Considering a single cell, the average delay per request $T_k$ at the steady state is derived as
\begin{align}
  T_k &= \frac{\rho_k}{\lambda_k} + \frac{\sum_{i=1}^3\frac{\lambda_{k,i}}{\mu_i^2}}{(1-\rho_k)} \nonumber\\
  &= \left. \tau_1(R_{k,1} + k_2R_{k,2} + k_3R_{k,3})\right. \nonumber\\
  &\left.+\frac{\lambda_k\tau_1^2(R_{k,1}+ k_2^2R_{k,2} + k_3^2R_{k,3})}{1-\lambda_k\tau_1(R_{k,1} + k_2R_{k,2} + k_3R_{k,3})}.\right. \label{k-delay}
\end{align}

\end{proposition}
\begin{proof}
Due to page limitations, we skip the proof in this paper. Please refer to the reference \cite{standardqueuing}.
\end{proof}
\subsection{Optimization Problem}
Based on the analysis of the average delay in a single cell, we can derive the sum average delay per request of the multi-cell system immediately as:
\begin{equation}\label{T}
  T = \frac{1}{\lambda}\sum_{k=1}^K\lambda_kT_k,
\end{equation}
where $\lambda = \sum_{k=1}^K\lambda_k$ denotes the overall user request arrival rate in the multi-cell network.

We can observe from (\ref{k-delay}) and (\ref{T}) that the average delay depends on the cache strategy. Because of the limited caching capacity, we would like to optimize the cache placement in BSs to minimize the sum average delay. Then the optimization problem is formulated as
\begin{align}
 &\min~\ \ \ \ \ \ \ \ \ T \nonumber\\
  & s.t.                                                                                                                                                                                                                                                                                                                                                                                                                                                                                                                                                                                                                                                                                                                                                                                                                                                                                                                                                                                                                                                                                                                                                                                                                                                                                                                                                                                                                                                                                                                                                                                                                                                                                                                                                                                                                                                                                                                                                                                                                                                                                                                                                                                                                                                                                                                                                                                                                                                                                                                                                                                                                                                                                                                                                                                                                                                                                                                                                                                                                                                                                                                                                                                                                                                                                                                                                                                                                                                                                                                                                                                                                                                                                                                                                                                                                                                                                                                                                                                                                                                                                                                                                                                                                                                                                                                                                                                                                                                                                                                                                                                                                                                                                                                                                                                                                                                                                                                                                                                                                                                                                                                                                                                                                                                                                                                                                                                                                                                                                                                                                                                                                                                                                                                                                                                                                                                                                                                                                                                                                                                                                                                                                                                                                                                                                                                                                                                                                                                                                                                                                                         \ \ \sum_{f=1}^F c_{k,f}S_f \leq C_k, \qquad \forall k \in \mathcal{K},\nonumber \\
  &\ \ \hspace{5mm}c_{k,f} \in \{0,1\}, \hspace{3mm}\qquad \quad\forall k \in \mathcal{K}, f \in \mathcal{F}.\label{T1}
\end{align}

The first constraint in (\ref{T1}) represents the storage constraints and the second one indicates the optimization variables are binary (i.e., caching or no caching), respectively.

\section{Approximation Algorithms}
In this section, we first show that the optimization problem (\ref{T1}) is NP-complete in a special case and then solve it via local greedy algorithm. For the general case, we formulate the optimization problem as the maximization of a monotone submodular function subject to matroid constraints, allowing us to adopt a conventional greedy algorithm within $1/2$ performance guarantee. Furthermore, we propose a heuristic greedy algorithm which is shown to achieve significant performance gain in the numerical results.

\subsection{A Special Case}
In this subsection, we consider the case where $k_2 = k_3$ and $\lambda_k = \lambda,\ \forall\ k$. That is, fetching requested contents from cooperative BSs requires the same delay as that from the core network via backhaul link and the user request rates are the same across all the cells.

In this case, $T_k$ is rewritten as:
\begin{align}\label{single}
  T_k{\setlength\arraycolsep{0.5pt}=} \tau_1[(1{\setlength\arraycolsep{0.5pt}-}k_3)R_{k,1} {\setlength\arraycolsep{0.5pt}+}k_3]+ \frac{\lambda\tau_1^2[(1-k_3^2)R_{k,1} + k_3^2]}{1-\lambda\tau_1[(1-k_3)R_{k,1} + k_3]},
\end{align}
where $R_{k,1} = \sum_{f = 1}^{F} \bar P_{k,f}c_{k,f}$ denotes the sum probability for requests to be served directly from associated BS $b_k$. Given the fact that $k_3 > 1$, we note that $T_k$ monotonically decreases with the increase of $R_{k,1}$. Moreover, $\{R_{k,1}, k = 1,\cdots,K\}$ at BSs are independent of each other. Hence, the original minimization content placement problem can be separated into $K$ independent maximization problems with the objective function $R_{k,1}$ respectively. Without loss of generality, we omit the index $k$ and the optimization problem at each BS is equivalent to:

\begin{align}\label{T2}
  &\max~ \hspace{6mm}  \sum_{f = 1}^{F} \bar P_{f}c_{f},  \\
  & s.t.                                                                                                                                                                                                                                                                                                                                                                                                                                                                                                                                                                                                                                                                                                                                                                                                                                                                                                                                                                                                                                                                                                                                                                                                                                                                                                                                                                                                                                                                                                                                                                                                                                                                                                                                                                                                                                                                                                                                                                                                                                                                                                                                                                                                                                                                                                                                                                                                                                                                                                                                                                                                                                                                                                                                                                                                                                                                                                                                                                                                                                                                                                                                                                                                                                                                                                                                                                                                                                                                                                                                                                                                                                                                                                                                                                                                                                                                                                                                                                                                                                                                                                                                                                                                                                                                                                                                                                                                                                                                                                                                                                                                                                                                                                                                                                                                                                                                                                                                                                                                                                                                                                                                                                                                                                                                                                                                                                                                                                                                                                                                                                                                                                                                                                                                                                                                                                                                                                                                                                                                                                                                                                                                                                                                                                                                                                                                                                                                                                                                                                                                                                                \hspace{8mm}\sum_{f=1}^F c_fS_f \leq C, \qquad \forall f \in \mathcal{F},\\
  &\hspace{14mm}c_f \in \{0,1\}, \qquad\quad \forall  f \in \mathcal{F}.
\end{align}

Obviously, the optimization problem is reduced to the well-known \textit{0-1} knapsack problem, which is already proved to be NP-complete \cite{knappack}. Hence, when $k_2 = k_3$, we can obtain the optimal solution via dynamic programming. However, in our case, considering the assumption that the content size is exponentially distributed while the dynamic programming is based on integer weights and values, we adopt the greedy algorithm instead of dynamic programming. Specifically, we first sort the content items in descending order according to their popularity-to-size ratio (i.e., $\frac{\bar P_f}{S_f}$). In each step, we greedily choose the remaining item with the largest popularity-to-size ratio until the cache at BS is filled up. Since each BS conducts the content placement independently, we call it as local greedy caching (LGC) strategy to distinguish with the following greedy algorithms.
\subsection{The General Case}
In this subsection, we formulate the optimization problem defined in (\ref{T1}) as a monotone submodular function subject to matroid constraints as follows.

\begin{lemma}
The constraints of (\ref{T1}) can be mapped into a partition matroid.
\end{lemma}

\begin{proof}
Based on the definition for the partition matroid in \cite{femto}, let $G_k = \{g_k^1,\cdots,g_k^F\}$ denote the set of all the contents potential to be cached at BS $b_k$. The element $g_k^f$ represents the placement action of caching content $f$ into BS $b_k$. In this way, we get $K$ disjoint sets, i.e., $G_1,\cdots,G_K$. Then we define the ground set $G$ as $G := G_1\cup\cdots\cup G_K =\{g_1^1,\cdots,g_1^F,\cdots,g_K^1,\cdots,g_K^F\}$.

For a given cache placement \textbf{C} denoted in the second constraint of (\ref{T1}), we map it into the set $X$ by including element $g_k^f$ into the set $X$ if and only if $c_{x,f} = 1$. Obviously, $X\subseteq G$. Accordingly, denote with $X \cap G_k$ the collection of contents that have been cached into BS $b_k$. Corresponding with the storage constraint of (\ref{T1}), the cardinality of the set $X\cap G_k$ should satisfy that $\mid X \cap G_k\mid\ \leq N_k$, where $N_k$ is deduced from the storage capacity of BS $b_k$. Then we define the collection of all the possible cache placement sets as:

\begin{equation}\label{I}
  \mathcal{I} = \{X\subseteq G: \mid X \cap G_k\mid\ \leq N_k, \forall k\}.
\end{equation}

Obviously, the structure of $\mathcal{I}$ is similar to that defined in \cite{femto}. Thus, the constraints (\ref{T1}) can be written as partition matroid, denoted by $\mathcal{M} = (G,\mathcal{I})$.
\end{proof}

\begin{lemma}\label{supermodular}
$T$ is monotonically nonincreasing supermodular set function.
\end{lemma}

\begin{proof} Given the mapping from the content placement \textbf{C} into the set $X$, we learn that the objective function $T$ is a set function defined on $X$.

For the monotonicity, we note that a single content will not increase the average delay when added into the placement. Hence, $T$ is a nonincreasing function with respect to $X$.

For the submodularity, we need to prove the increasing return property for $T$, i.e., for any two content placement sets $X \subseteq Y \subseteq G, T(X\cup g_k^f)- T(X) \leq T(Y\cup g_k^f) - T(Y)$. Considering the closedness property for supermodular function that non-negative linear combination of supermodular functions is still supermodular, it is sufficient for us to prove the supermodularity of $T_k$.

For convenience of analysis, we rewrite $T_k$ as:

\begin{equation}\label{Tk2}
  T_k = \frac{1}{\lambda_k}\rho_k + f_kg(\rho_k),
\end{equation}
where $\rho_k = \lambda_k\tau_1(R_{k,1} + k_2R_{k,2} + k_3R_{k,3})$, $f_k = \lambda_k\tau_1^2(R_{k,1} + k_2^2R_{k,2} + k_3^2R_{k,3})$ and $g(\rho_k) = \frac{1}{1-\rho_k}$, respectively.

Recalling the closedness property of supermodular function, we only need to prove the supermodularity of $\rho_k$ and $f_kg_k$ respectively. First, we focus on the proof of supermodularity of $\rho_k$. We define the decremental value of $R_{k,i}, i = 1,2,3$ and $\rho_k$ that a new element $g_j^f$ makes when included into the content placement set $X$ as:
\begin{equation}\label{differencer}
  \Delta R_{k,i}(X) = R_{k,i}(X\cup g_j^f) - R_{k,i}(X), i = 1,2,3,
\end{equation}
\begin{align}\label{differencer}
  &\Delta \rho_k(X) = \rho_k(X\cup g_j^f) - \rho_k(X) \nonumber \\
  &\hspace{12mm}=\lambda_k\tau_1(\Delta R_{k,1}+k_2 \Delta R_{k,2}+k_3 \Delta R_{k,3}).
\end{align}

Notice that the decremental value of $\rho_k$ that the new element $g_j^f$ makes differs with respect to the content placement \textbf{X}. For simplicity, we denote $\bar P_{k,f}$ as $P_f$. With in mind that $R_{k,i}$ denotes the probability for requests served via the \textit{i-th} route, we categorize \textbf{X} according to $g_j^f$ as follows: 1) $j = k$, i.e., caching content $f$ at BS $b_k$. In this regard, $\Delta R_{k,1}(X) = P_f$. If the content $f$ has not been cached in all BSs, i.e., $\sum_{k=1}^K x_{k,f} = 0$, denoting the corresponding set as $X_1 := \{X \mid \sum_{i=1}^K x_{i,f} = 0\}$, we get $\Delta R_{k,2}(X_1) = 0$ and $\Delta R_{k,3}(X_1) = -P_f$. Hence, $\Delta \rho_{k,1}(X_1) = \lambda_k\tau(1-k_3)P_f$. Otherwise, we define the corresponding content placement set $X_2$ as $X_2 := \{X \mid x_{k,f} = 0, max_{i \neq k} x_{i,f} = 1\}$ and get $\Delta R_{k,2}(X_2) = -P_f$, $\Delta R_{k,3}(X_2) = 0$. Thus, $\Delta \rho_{k,2}(X_2) = \lambda_k\tau(1-k_2)P_f$; 2) $j \neq k$. In this respect, $\Delta R_{k,1}(X) = 0$ and the content placement set $X$ is classified into three types. If BS $b_k$ has already cached content $f$, i.e., $x_{k,f} = 1$, the new element $g_j^f$ will not change the value of $\rho_{k}$ since requests for content $f$ at $b_k$ always get satisfied directly from the cache of $b_k$. Denoting the corresponding content placement set as $X_3 := \{X \mid x_{k,f} = 1\}$, we have $\Delta \rho_{k,3}(X_3) = 0$. Otherwise, if other cells have not cached content $f$ neither, indexed as $X_4$, we get $\Delta R_{k,2}(X_4) = P_f$ and $\Delta R_{k,3}(X_4) = -P_f$. Hence, $\Delta \rho_{k,4}(X_4) = \lambda_k\tau(k_2-k_3)P_f$. However, if there is any BS that has already cached content $f$, the new element $g_j^f$ will not change $\rho_k$ since requests for content $f$ at BS $b_k$ always get satisfied from its cooperative BSs. Denote this kind of content placement as $X_5$ and then we has $\Delta \rho_k(X_5)=0$.

We summarize the above analysis as following:
\begin{equation}\label{rou}
 \Delta \rho_{k}(X) = \\
 \begin{cases}
 \lambda_k\tau(1-k_3)P_f & X \in {X_1}\\
 \lambda_k\tau(1-k_2)P_f & X \in {X_2}\\
 \lambda_k\tau(k_2-k_3)P_f & X \in {X_4}\\
    0  &   X \in {X_i, i = 3,5} \\
 \end{cases}
\end{equation}

Considering $1\leq k_2 \leq k_3$, we get $\Delta \rho_{k}(X) \leq 0, \forall X$. Hence, $\rho_{k}(X)$ is a nonincreasing function.

Based on (\ref{rou}), $\forall X \subseteq Y \subseteq S$, $\Delta \rho_{k}(Y)-\Delta \rho_{k}(X)$ is
\begin{align}\label{differrou}
 \Delta  \rho_k(Y)-& \Delta \rho_{k}(X) = \nonumber\\
 &\left\{
 \begin{array}{l}
   0 \qquad\qquad\qquad\ \ \ \ X,Y \in {X_i, \forall i}\\
    \qquad\qquad\ \ \  \ \ \ \text{or}\  X \in {X_5}, Y \in {X_3}\\
 \lambda_k\tau(k_3-k_2)P_f  \ \  X \in {X_1}, Y \in {X_2}
 \\ \qquad\ \  \text{or}\  X \in {X_4}, Y \in {X_i, i = 3,5}\\
 \end{array}
 \right.
\end{align}

From above analysis, $\forall X \subseteq Y \subseteq S$, we obtain that $\Delta \rho_{k}(Y)-\Delta \rho_{k}(X)\geq 0$. Therefore, $\rho_k$ is a monotonically nonincreasing supermodular function. Noting the similar structure of $f_k$ to that of $\rho_k$, we get the nonincreasing property and supermodularity of $f_k$ immediately.

Then, for $g(\rho_{k}) = \frac{1}{1-\rho_{k}}$, we derive its first-order derivative and second-order derivative versus $\rho_k$ respectively as:
\begin{equation}\label{rouderive}
\frac{dg(\rho_{k})}{d\rho_k} = \frac{1}{(1-\rho_k)^2} ,
\end{equation}
\begin{equation}\label{rouderive2}
\frac{d^2g(\rho_{k})}{d\rho_k^2} = \frac{2}{(1-\rho_k)^3} .
\end{equation}

Considering the stability condition $\rho_k < 1$, we get $\frac{dg(\rho_{k})}{d\rho_k} > 0$ as well as $\frac{d^2g(\rho_{k})}{d\rho_k^2} > 0$. Thus, $g(\rho_{k})$ has the same monotonicity as that of $\rho_k$ and is a convex function of $\rho_k$. Therefore, $g(\rho_{k})$ is also a monotonically nonincreasing supermodular function \cite{learninging}. From \cite{Boyd}, we get that $f_kg(\rho_{k})$ is also supermodular.

Hence, $T_k$ is a monotonically nonincreasing supermodular function and the lemma is proved.
\end{proof}

\begin{algorithm}[t!]
\caption{Heuristic Greedy Algorithm}
\label{Algorithm2}
\begin{algorithmic}[1]
\STATE \textbf{Input}. $K,F,\{S_f, \forall f\},\{C_k, \forall k\}$
\STATE \textbf{Output}. Cache placement \textbf{C}
\STATE \textbf{Initialize}. \textbf{C} $\leftarrow$ $(0)_{K\times F}$, $\mathcal{D} \leftarrow \mathcal{K\times F}$, remaining storage size for BS $k$, $F_k = C_k$.
\WHILE{$\mathcal{D}\neq \varnothing$}
\STATE  $(k^*,f^*) \leftarrow argmax_{(k,f) \in \mathcal{D}} \frac{ T(\textbf{C})-T(\textbf{C} \cup d_{k,f})}{S_f}$
\STATE  $c_{k^*,f^*} = 1$
\STATE  $F_{k^*} = F_{k^*} - S_{f^*}$
\STATE  $\mathcal{D} \leftarrow \mathcal{D} \setminus(k^*,f^*)$
 \FOR   {$k = 1:K $}
\FOR{$f = 1:F$}
\IF {$S_f > F_k$}
\STATE $\mathcal{D} \leftarrow \mathcal{D} \setminus(k,f)$
\ENDIF
\ENDFOR
\ENDFOR
\ENDWHILE
\end{algorithmic}
\end{algorithm}

Based on Lemma 2, we directly get that $-T$ is a nondecreasing submodular function. Hence, the optimization problem defined in (\ref{T1}) is equivalent to the maximization of a nondecreasing submodular function subject to matroid constraints, denoted as:
\begin{align}\label{T2}
 &\max~\ \ \ \ \ \ \ \ \ -T \\
  & s.t.                                                                                                                                                                                                                                                                                                                                                                                                                                                                                                                                                                                                                                                                                                                                                                                                                                                                                                                                                                                                                                                                                                                                                                                                                                                                                                                                                                                                                                                                                                                                                                                                                                                                                                                                                                                                                                                                                                                                                                                                                                                                                                                                                                                                                                                                                                                                                                                                                                                                                                                                                                                                                                                                                                                                                                                                                                                                                                                                                                                                                                                                                                                                                                                                                                                                                                                                                                                                                                                                                                                                                                                                                                                                                                                                                                                                                                                                                                                                                                                                                                                                                                                                                                                                                                                                                                                                                                                                                                                                                                                                                                                                                                                                                                                                                                                                                                                                                                                                                                                                                                                                                                                                                                                                                                                                                                                                                                                                                                                                                                                                                                                                                                                                                                                                                                                                                                                                                                                                                                                                                                                                                                                                                                                                                                                                                                                                                                                                                                                                                                                                                                         \ \ \sum_{f=1}^F c_{k,f}S_f \leq C_k, \qquad \forall k \in \mathcal{K},\label{C3} \\
  &\ \ \hspace{5mm}c_{k,f} \in \{0,1\}, \hspace{3mm}\qquad \quad\forall k \in \mathcal{K}, f \in \mathcal{F}.\label{C4}
\end{align}

A conventional way to solve such maximization problem is via greedy algorithm with $1/2$ performance guarantee \cite{submodular}. Specifically, starting with empty caches, at each step we greedily choose the file that maximizes marginal value (i.e., $T(\textbf{C})-T(\textbf{C} \cup d_{k,f})$) into the cache placement until all the caches are filled. When the marginal value is zero, the algorithm will also stop. Hence, there would be $\lambda_sCK$ iterations on average if all the BSs are of the same storage capacity $C_k = C$. Each iteration consists of evaluating marginal value of no more than $KF$ elements. Each evaluation takes $O(K)$ time. Therefore, the running time for the greedy algorithm would be $O(\lambda_sCFK^3)$. We refer to this conventional greedy caching algorithm as CGC strategy.

Given the the size diversity among contents and the limited storage size, we propose a heuristic greedy caching (HGC) policy illustrated in Algorithm \ref{Algorithm2}, where we evaluate the value of cache element $g_k^f$ by the ratio of the marginal value it brings to its content size. When $k_2 = k_3$, given the fact that $T$ monotonically decreases with $R_1$, we can see our proposed HGC policy would have the same performance as LGC strategy, while the CGC strategy would only achieve the performance of most popular caching scheme (MPC) where we greedily choose the most popular files into the caches. 
\begin{table}[t!]
\renewcommand\arraystretch{1.7}
\centering
\caption{SIMULATION PARAMETERS}\label{table2}
\begin{tabular}{l|c}
  \hline
  \textbf{Parameter} & \textbf{Parameter Value}\\
  \hline
  $F$: \textit{Content\  Number} & $100$  \\
  \hline
 \textit{$\bar S$: Content\  Average\  Size } & $5\ M bits$ \\
  \hline
  \textit{$r_1$: Wireless \ Transmission \ Rate }& $100\ M bps$ \\
  \hline
  $\ \ \ \ \ \ \ \ k_2=\tau_2/\tau_1$ & $4$\\
  \hline
  $\ \ \ \ \ \ \ \ k_3=\tau_3/\tau_1$ & $20$\\
  \hline
  $\lambda_k$: \textit{Request\ Arrival\ Rate}&$0.5\ requests/s $\\
  \hline
  $\gamma$: \textit{popularity parameter} & $0.5$ \\
  \hline
\end{tabular}
\end{table}
\section{Numerical Results}
The performance of the proposed HGC strategy is evaluated in this section. The aforementioned MPC scheme, LGC strategy and CGC policy are also assessed as the performance benchmarks. The typical parameter settings used in the simulation and calculation are shown in Table \ref{table2}.

We first evaluate the accuracy of the derived analytical expressions in Propositon $1$ under the MPC strategy. The simulation results are presented along with the analytical ones in Fig. \ref{size}. It can be seen that the theoretical results are in excellent match with the simulation results. Intuitively, the average delay decreases as the parameter $\gamma$ goes up. The impacts of the network resources, the user request traffics and the transmission capacity among cells on the system performance are then discussed as following.


\textbf{Joint impact of the storage size and the number of cooperative BSs.} In Fig. \ref{CELLNUMBER}, we evaluate the combined effects of the cache sizes and the number of cells on the performance achieved from our proposed HGC algorithm. As expected, $-T$ monotonically increases with the increase of the storage size and the cell number. That is, the larger the cache sizes and the cell number are, the less the average delay is required. However, as one can see, the performance gain is diminishing as the cache sizes and the number of cells increase, which reconfirms the submodularity of $-T$. We also observe that when the network has small caching capability, the increase of the number of collaborative BSs only achieves trivial delay gains, whereas the performance gain is more significant when the cache size becomes larger. This result indicates that \emph{caching facilitates the exploitation of the freedom offered from the cooperation among BSs.}

\begin{figure}
\centering
\includegraphics[width=3.2in]{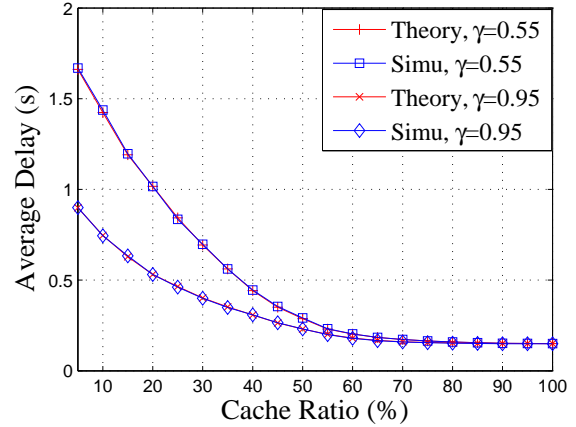}
\caption{Validation of theoretical analysis. We assume each BS has the same storage capacity (i.e., $C_k=C$). Cache ratio ($\%$) means the ratio of the storage size ($C$) to the average sum size of all the contents ($\frac{C}{\bar SF}$). }\label{SRvsNtwithdiffLI} \label{size}
\end{figure}

\begin{figure}
\centering
\includegraphics[width=3.2in]{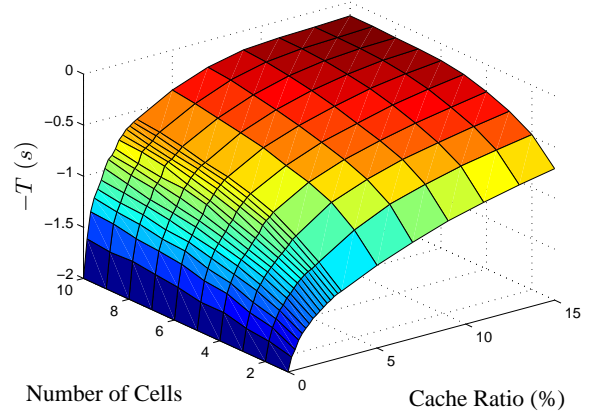}
\caption{Impact of the number of cells and the caching capability.}\label{CELLNUMBER}
\end{figure}

\textbf{Impact of the user request arrival rate.}
We evaluate the impact of the user request arrival rate on the average delay in Fig. \ref{requestdiffer1}. We consider the scenarios where the request arrival rates are heterogeneous across different cells. Here, we set $K=3$, $C = 50\ Mbits$, $F = 100$ and $\lambda_{1}=\lambda_{2}=0.05\ requests/s$. Altering $\lambda_{3}$ from $0.05$ to $1\ requests/s$, we can see that the delay for each content caching strategy increases with $\lambda_{3}$, since larger request rate increases the probability of longer waiting time for each request. It is also observed that the collaborative techniques achieve significant performance gains over the non-cooperative scheme. Our proposed algorithm HGC consistently outperforms the MPC, LGC and CGC schemes, with the gains increasing with $\lambda_{3}$ (up to $65 \%$, $50 \%$ and $33 \%$, respectively). Thus, the HGC policy is much more robust with the increase of user request rates and takes more advantage of the traffic diversity among cells. That is, \emph{the heavier and more diverse the traffic load is, the larger benefit the proposed HGC scheme brings.}



\textbf{Impact of the transmission capacity among BSs.} In Fig. \ref{trans}, we explore how the transmission capacity among BSs affects the performance. As mentioned in Section II, we denote with $k_i$ the transmission delay ratio between $i$-th Route and Route 1, where $1 <k_2 <k_3$. In particular, we keep $k_3 = 20$ constant and change the value of $k_2$ from $1$ to $20$. We see that the average delay increases with $k_2$ since the delay incurred by the service from collaborative BSs becomes larger. Again, our proposed HGC strategy outperforms the other three schemes. As expected, when $k_2 = k_3$, HGC strategy achieves the same performance as the LGC strategy, whereas the performance of CGC scheme reduces to that of the MPC policy. 

\section{CONCLUSION}
In this paper, we take a queuing theoretical approach to analyze the average delay per request in the cache-enabled multi-cell cooperative network. In this way, it allows us to optimize the cooperative caching scheme by taking into account not only the content popularity and storage capacity but also the user request stochastic traffic. By formulating the optimization problem as the maximization of a nondecreasing submodular function subject to matroid constraints, we are able to solve it via the CGC algorithm within a factor of $1/2$ of the optimum. We further propose the HGC algorithm which is shown to achieve significant performance gain compared with the CGC strategy. 
Finally, the valuable insights on the impacts of the key network parameters, the user request traffics and the transmission capacity among cells on the system performance are summarized as following:
\begin{itemize}
  \item The cooperation among cells provides significant delay gains over the non-cooperative scenario.
  \item Content caching facilitates the exploitation of the freedom offered from the cooperation among cells.
  \item The heavier and more diverse the traffic load is, the larger benefit the proposed HGC scheme brings.

\end{itemize}
%
%
%
%
%
%
%
%
%
%
%
%
%
%
\begin{figure}
\centering
\includegraphics[width=3.2in]{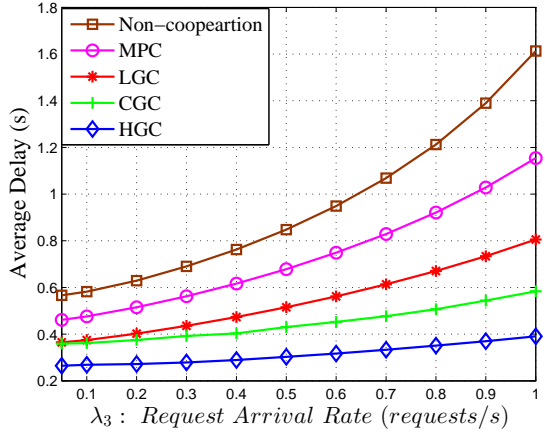}
\caption{Impact of user request arrival rate.}\label{requestdiffer1}
\end{figure}
\begin{figure}
\centering
\includegraphics[width=3.2in]{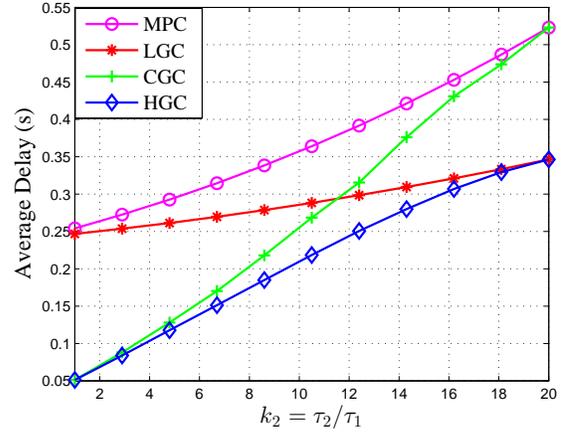}
\caption{Impact of transmission capacity among cells when $K=3$ and cache ratio equals to $40\%$.}\label{trans}
\end{figure}

\bibliographystyle{IEEEtran}
\bibliography{paper}

\end{document}